\RequirePackage{amsmath}
\documentclass[twoside,11pt]{article}
\usepackage{jmlr2e}

\usepackage[utf8]{inputenc} 
\usepackage[T1]{fontenc}    
\usepackage{hyperref}       
\usepackage{url}            
\usepackage{booktabs}       
\usepackage{amsfonts}       
\usepackage{nicefrac}       
\usepackage{microtype}      

\usepackage{graphicx,multicol}
\usepackage{amsfonts}
\usepackage{algorithm}
\usepackage{algorithmic}
\usepackage{amssymb}
\usepackage{relsize,bm}
\usepackage{booktabs}
\usepackage{hyperref}  
\usepackage{mathtools}
\newcommand{\R}{\mathbb{R}}
\newcommand{\E}{\mathbf{E}}
\newcommand{\I}{\bm{I}}
\newcommand{\1}{\mathbf{1}}
\renewcommand{\vec}[1]{{\bm{#1} }}
\newcommand{\ddtheta}{\,\mathrm{d}\vec{\theta}}
\newcommand{\dtheta}{\,\mathrm{d}\vec{\theta}}

\newcommand{\dpp}{\,\mathrm{d}\vec{p}}
\newcommand{\dt}{\,\mathrm{d}t}
\newcommand{\dW}{\,\mathrm{d}\vec{W}}

\newcommand{\LL}{{\mathcal L}}
\newcommand{\LA}{{\mathcal A}}
\newcommand{\NOGIN}{{\emph{NOGIN}} }

\newcommand{\cov}{\text{Cov}}
\newcommand{\N}{\text{N}}
\newcommand{\grad}{\vec{\nabla}}
\newcommand{\T}{{\mathsmaller{\textrm T}}}

\jmlrheading{1}{2000}{1-48}{4/00}{10/00}{xxxx00x}{Charles Matthews and Jonathan Weare}

\ShortHeadings{Langevin Markov Chain Monte Carlo with stochastic gradients}{Matthews and Weare}
\firstpageno{1}

\begin{document}

\title{Langevin Markov Chain Monte Carlo with stochastic gradients}

\author{\name Charles Matthews \email c.matthews@ed.ac.uk \\
       \addr Department of Mathematics, The University of Edinburgh
       \AND
       \name Jonathan Q. Weare \email weare@nyu.edu \\
       \addr Courant Institute of Mathematical Sciences, New York University}

\editor{TBA}

\maketitle

\begin{abstract}%
Monte Carlo sampling techniques have broad applications in machine learning, Bayesian posterior inference, and parameter estimation. Often the target distribution takes the form of a product distribution over a  dataset with a large number of entries.  For sampling schemes utilizing gradient information it is cheaper for the derivative to be approximated using a random small subset of the data, introducing extra noise into the system. We present a new discretization scheme for underdamped Langevin dynamics when utilizing a stochastic (noisy) gradient. This scheme is shown to bias computed averages to second order in the stepsize while giving exact results in the special case of sampling a Gaussian distribution with a normally distributed stochastic gradient.
\end{abstract}

\section{Introduction} \label{sec::intro}
A commonly encountered problem in data science and machine learning applications is the sampling of parameters $\vec{\theta} \in \R^D$ from a product distribution over a dataset $\vec{y}$ containing $N$ observations $\vec{y}_i$. Combined with a regularizing prior distribution $\pi_0(\vec{\theta})$, the aim is to generate points distributed according to the target distribution $\pi(\vec\theta)$ where 
\begin{equation}\label{eqn::target_distribution}
\pi(\vec{\theta}) \propto \pi_0(\vec{\theta})\prod_{i=1}^N \pi(\vec{y}_i\,|\,\vec\theta).
\end{equation}
Such a formulation is commonplace in  Bayesian inverse applications, where $\pi(\vec{\theta})$ is referred to as a posterior distribution.  Markov Chain Monte Carlo (MCMC) schemes are an   effective method for sampling from such distributions \citep{robert2004monte,brooks2011handbook}, however conventional schemes (see e.g. \citep{metropolis,neal2011mcmc}) require an evaluation of $\grad\log\pi(\vec\theta)$ or $\pi(\vec\theta)$ (up to a normalization constant) at each step to propose new points. Because each evaluation requires a full pass through the dataset this can be prohibitively expensive for  large $N$, and therefore there is considerable interest in MCMC schemes that   sample the target distribution but do not require   access to the entirety of $\vec y$.

In this article we are interested in sampling $\pi$ using stochastic approximations to the  gradient. For a distribution as in \eqref{eqn::target_distribution},  the gradient vector (or the {\emph{force}})  is a sum over the data
\begin{equation} \label{eqn::true_force}
\vec F(\vec \theta) := \grad\log\pi_0(\vec\theta) + \sum_{i=1}^N \grad\log\pi(\vec{y}_i\,|\,\vec\theta)
\end{equation}
where $\grad$ takes a gradient with respect to $\vec\theta$. Gradient information is used to drive proposals for an MCMC scheme in efficient directions in the $D$ dimensional space. We will refer to $\vec F$ as the exact or `true' force vector, using all $N$ datapoints to compute the derivative. If $N$ is large it may become necessary to instead use a computationally cheaper approximation to the true gradient. In this article we will consider a \emph{noisy gradient} estimator $\widetilde{\vec F}(\vec \theta)$ with the properties 
\begin{equation}\label{eqn::noisyf}
\E\left[\widetilde{\vec F}(\vec \theta)\right] = \vec{F}(\vec\theta),\qquad \cov\left[\widetilde{\vec F}(\vec \theta)\right] = \vec{\Sigma}(\vec\theta) 
\end{equation}
for positive semi-definite $\vec{\Sigma}$. There are many choices of  estimator in the literature, with much recent work undertaken to increase  accuracy (quantified by the size of $\vec{\Sigma}$) and reduce storage requirements (see e.g.  \citep{baker2017control,dubey2016variance}).  Usually the primary cost of the estimator comes from an evaluation of force terms over a random mini-batch $\widetilde{\vec{y}}\subseteq \vec{y}$ where $|\widetilde{\vec{y}}|=n$ for some fixed $n\leq N$. The mini-batch is redrawn from $\vec{y}$ at every estimation of $\widetilde{\vec{F}}$, with $\vec{\Sigma}\to\vec{0}$ as $n\to N$ if the redrawing is without replacement. The simplest estimator choice is using  a scaled sum over a mini-batch uniformly subsampled from $\vec{y}$ \citep{robbins1951stochastic}:
\begin{equation} \label{eqn::robmun}
\widetilde{\vec F}(\vec \theta) := \grad \log\pi_0(\vec \theta)+ \frac{N}{n} \sum_{i=1}^n \grad\log\pi(\tilde{\vec{y}}_i\,|\,\vec\theta), \quad 
\vec{\Sigma}(\vec \theta) := \frac{N(N-n)}{n}\cov\left[\{\vec{F}_i(\vec\theta)\}\right]
\end{equation}
where $\vec{F}_i(\theta) := \grad\log\pi(\vec{y}_i\,|\,\vec\theta)$. An estimate for $\vec{\Sigma}$  can be obtained in practice by computing the covariance of $n$-many $\vec{F}_i$ terms taken over a mini-batch. Such approximations are central to many   algorithms' attempts to reduce the observed error, as the additional stochastic  term introduced into the MCMC proposals can lead to a large bias if not accounted for. The size of this bias  and the rate of decorrelation of samples largely differentiates noisy gradient methods, and is the subject of the present article.

\begin{algorithm}[tb]
    \caption{\NOGIN: Noisy Gradient Integrator}
    \label{alg::nogin}
\begin{algorithmic}[1]
	\REQUIRE $\vec{\theta}_0,\,h>0,\,\gamma>0,\,T>0$
	\STATE \textbf{Initialize}: $\,\,$ $\vec p \sim \N(\vec 0,\I)$;  $\,\vec\theta \leftarrow \vec\theta_0$; $\,\lambda \leftarrow \sqrt{\tanh(\gamma h /2) }$
    \FOR{$t=1$ to $T$}
	\STATE $\vec\theta\leftarrow \vec\theta  + \tfrac{h}2 \vec p $
	\STATE $\widetilde{\vec F} \leftarrow \widetilde{\vec F}(\vec \theta )$, $\,\,\vec\Sigma \leftarrow {\cov}(\widetilde{\vec F}(\vec \theta ))$, $\,\,\vec R  \sim \N(\vec 0, \I)$
	\STATE $\vec p \leftarrow  \vec p + \tfrac{h}2 \widetilde{\vec F} +  {\lambda}\vec R $
	\STATE $\vec p \leftarrow  \left((1-\lambda^2)\I-\frac{h^2}4\vec \Sigma \right) \left((1+\lambda^2)\I+\frac{h^2}4\vec \Sigma \right)^{-1} \vec p$  \label{step::damping}
	\STATE $\vec p \leftarrow  \vec p + \tfrac{h}2 \widetilde{\vec F} + {\lambda}\vec R $
	\STATE $\vec \theta \leftarrow  \vec \theta + \tfrac{h}2\vec p$
	\STATE $\vec \theta_t \leftarrow  \vec \theta$
	\ENDFOR  
\end{algorithmic}
\end{algorithm}
In the sequel we assume that we are given some estimator for the force $\widetilde{\vec{F}}$ and are able to compute an estimate for its covariance $\vec{\Sigma}$. We present a new sampling scheme that we refer to as the Noisy Gradient Integrator (or \NOGIN$\!\!$), given in Algorithm \ref{alg::nogin}. The structure of the article is as follows. In Section \ref{sec::noisygrad} we set notation for the noisy schemes and place our method in a proper context. Section \ref{sec::nogin} derives our proposed \NOGIN scheme from established Langevin dynamics methods. Section \ref{sec::analysis} gives some analytical results for the expected error from the \NOGIN scheme. In Section \ref{sec::experiments} we compare the proposed method against others on a Bayesian inference application. A discussion of results, outlook and ramifications concludes the article in Section \ref{sec::conclusion}.

\section{Noisy Gradient Integration}\label{sec::noisygrad}
An effective way to utilize gradient information in MCMC schemes is by proposing points from solution trajectories of ergodic dynamics that sample the target distribution $\pi$ \citep{brooks2011handbook}. For appropriate test functions $f$, ergodicity implies that for such a solution trajectory $\vec{\theta}(t)$
\[
\lim_{T\to\infty} \frac{1}{T}\sum_{t=1}^T f(\vec{\theta}(t)) = \int f(\vec{\theta}) \pi(\vec{\theta}) \ddtheta =: \langle f \rangle.
\]
However exact solutions for dynamics involving gradients of complicated, nonlinear $\pi$ are seldom known. A discretization scheme is usually employed to advance the state through time by computing a sequence of $\vec{\theta}_k\approx\vec{\theta}(kh)$ for a discretization timestep $h>0$.  
The discretization introduces statistical error into the computed trajectory at both finite and infinite time (see e.g. \citep{leimkuhler2015computation}). The infinite-time (or sometimes called asymptotic or perfect) sampling bias introduced by the schemes is the difference in the respective averages of $f$, 
\[
\text{Asymptotic Bias } = \left|\langle f \rangle_h - \langle f \rangle\right|,\qquad \langle f \rangle_h := \lim_{T\to\infty} \frac{1}{T}\sum_{t=1}^T f(\vec{\theta}_t),
\]
where $\langle f \rangle_h$ is the infinite-time observed average evaluated from numerical computation. 
The size of this bias can be tuned by decreasing the stepsize $h$, at the cost of sacrificing the rate of exploration for the scheme (see Section \ref{sec::zerosum}).  Alternatively schemes such as MALA \citep{roberts1996} employ an additional Metropolis-Hastings (MH) step to correct for the bias introduced from the discretization, at the cost of rejecting some moves.   

In the case of a target distribution in the form of \eqref{eqn::target_distribution}, the $O(N)$ cost of the force $\vec{F}$ or log-likelihood $\log\pi(\vec{\theta}\,|\,\vec{y})$  may be prohibitively expensive in terms of memory or computation. 
Na\"ively exchanging $\vec{F}$ for a stochastic estimate $\widetilde{\vec F}$ in an MCMC algorithm usually 
leads to the introduction of a large  bias thanks to the additional unaccounted noise term. 

We consider schemes using a fixed timestep parameter $h$, differing from decreasing-stepsize schemes such as SGLD   \citep{welling2011bayesian}, which reduce the bias at the cost of decreased computational efficiency as $T$ gets large \citep{nagapetyan2017true}. Many similar strategies involve either decreasing $h$ only  up to a given value (e.g. \citep{chen2014stochastic}) or running with a fixed stepsize and accepting the error introduced in exchange for faster convergence (see e.g. \citep{vollmer2016exploration,nagapetyan2017true}). 
The SGNHT scheme \citep{ding2014bayesian} is able to correct for the gradient noise without specific evaluation of $\vec{\Sigma}$ by extending the space to include one or more seperate thermostatting variables, as long as the covariance of the stochastic gradient $\vec\Sigma$ is independent of $\vec{\theta}$.  The  CCADL scheme \citep{shang2015covariance} builds upon the thermostatting framework of the SGNHT scheme  and proposes a first-order scheme that aims to reduce the bias in systems with stochastic gradients with non-constant covariance.

A modified version of the SGLD scheme (mSGLD) using fixed stepsize was introduced in \citep{vollmer2016exploration} and varies the strength of the introduced white noise term to balance against the nuisance noise from the stochastic gradient estimate.  Similarly SGHMC in \citep{chen2014stochastic} scales the variance of the Langevin noise term to offset the extra introduced noise, but uses a formulation that includes momentum.
While  alternatives to stochastic gradient methods exist for reducing the computational cost of the gradient over the data (see e.g. \citep{chen2015convergence,korattikara2014austerity,maclaurin2014firefly}) we shall keep our focus on a noisy gradient formulation in what follows.

Our aim in this article is to derive a new numerical scheme that introduces a bias of order $h^2$ when using a fixed stepsize  with a noisy gradient. Our strategy is similar to the aforementioned articles: we treat the stochasticity of the noisy gradient as an additional random noise term that acts to `heat' the resulting dynamics. However we build a discretization scheme in a way that allows us to 
incorporate the extra random term into an exact solve of an Ornstein-Uhlenbeck process already present in Langevin dynamics. Assuming we  have  some knowledge of the estimator being used, specifically the covariance of the noise $\vec{\Sigma}(\vec{\theta})$, we are able to appropriately damp the dynamics. We show numerically that this scheme has favorable properties even when the noise's covariance can only be approximated. We shall not consider correcting for the effects from moments higher than the covariance,  as   higher moments   introduce bias at   orders of $h$ beyond the usual order of our integrators, and hence will be `invisible' in error analysis. 

\section{The NOGIN Scheme}\label{sec::nogin}
\subsection{Langevin dynamics}
%
We will demonstrate that the \NOGIN scheme given in Algorithm \ref{alg::nogin} discretizes the underdamped Langevin dynamics SDE
\begin{equation}\label{eqn::langevindyn}
\dtheta = \vec p\dt,\qquad 
\dpp = \vec{F}(\vec \theta)\dt -\vec\mu(\vec\theta) \vec p \dt + \sqrt{2\vec\mu(\vec \theta)} \dW
\end{equation}
to second order, where $\vec\mu(\vec\theta)$ is a particular symmetric positive definite matrix  and $\vec{W}(t)$ is a standard Wiener process. 
This dynamics  uniquely preserve  the augmented distribution
\begin{equation}
\pi(\vec\theta,\vec p) := \pi(\vec\theta ) \N(\vec p\,|\,\vec{0},\I),\label{eqn::target}
\end{equation}
whose marginal in $\vec \theta$ yields the required target distribution.  We consider integration schemes built by additively decomposing the vector field of \eqref{eqn::langevindyn}  into three pieces, following the procedure in  \citep{leimkuhler2012rational}. The pieces, labeled A, B, and O, are defined as
\begin{align}\label{eqn::ldpieces}
{\textrm{d}} \left[\begin{array}{c}\vec \theta\\\vec p\end{array}\right] = 
\underbrace{\left[\begin{array}{c}\vec p\dt\\\vec 0\end{array}\right]}_{\textrm A}
+ 
\underbrace{\left[\begin{array}{c}\vec0\\ \vec{F}(\vec \theta)\dt\end{array}\right]}_{\textrm B} 
+
\underbrace{\left[\begin{array}{c}\vec0\\ -\vec\mu(\vec\theta) \vec p \dt + \sqrt{2\vec\mu(\vec\theta)} \dW\end{array}\right]}_{\textrm O}.
\end{align}
The A and B parts are referred to as the `drift' and `kick' pieces respectively, while the O `fluctuation' piece corresponds to an Ornstein-Uhlenbeck (OU) linear SDE process. Note that each of the pieces, when taken individually, has an explicit weak solution. If we define propagation functions
\begin{gather}
\Phi^{\text A}_{h}(\vec\theta,\vec p) := (\vec\theta + h \vec p , \vec p ), \qquad
\Phi^{\text B}_{h}(\vec \theta,\vec p) := (\vec\theta , \vec p + h \vec{F}(\vec\theta) ),\label{eqn::stepB}\\
\Phi^{\text O}_{\vec M,\vec R}(\vec\theta,\vec p) := \left(\vec\theta , \vec{M} \vec p + \sqrt{\I-\vec{M}^2 }\vec R\right)\label{eqn::stepO}
\end{gather}
for symmetric $\vec{M}\in\R^{D\times D}$ and vector $\vec{R}\in\R^D$, then for suitable test function $f$, denoting the backward Kolmogorov operator corresponding to piece X in \eqref{eqn::ldpieces} as $\LL_X$,
\begin{gather}
\left(e^{h \LL_{\text A}} f\right)(\vec{\theta},\vec{p}) = \E\left[ f\left(\Phi^{\text A}_{h}(\vec\theta,\vec p)\right)  \right],\qquad
\left(e^{h \LL_{\text B}} f\right)(\vec{\theta},\vec{p}) = \E\left[ f\left(\Phi^{\text B}_{h}(\vec\theta,\vec p)\right)  \right],\label{eqn::map1}\\
\left(e^{h \LL_{\text O}} f\right)(\vec{\theta},\vec{p}) = \E\left[ f\left(\Phi^{\text O}_{\vec{M}_h,\vec R}(\vec\theta,\vec p)\right)  \right] \quad {\text{if}} \,\,\vec R \sim  \N(\vec0,\I) \,\,{\text{and}} \,\,\vec{M}_h=\exp(-\vec\mu(\vec\theta)h).\label{eqn::map2}
\end{gather}
Numerical schemes for \eqref{eqn::langevindyn} can then be proposed by composing these three mappings in a prescribed sequence, using the A-B-O alphabet. For example, a step of the ABOBA scheme in \citep{leimkuhler2012rational} with stepsize $h>0$  and using $\vec\mu(\vec\theta)=\gamma\I$ (for friction constant $\gamma>0$)  can be written as
\[
(\vec\theta_{t+1},\vec p_{t+1}) = \Phi^{\text A}_{h/2}  \circ\Phi^{\text B}_{h/2}  \circ \Phi^{\text O}_{e^{-\gamma h\I},\vec R_t} \circ \Phi^{\text B}_{h/2}  \circ \Phi^{\text A}_{h/2}(\vec \theta_t,\vec p_t).
\]
This labeling convention provides a family of \emph{splitting} methods that integrate the dynamics \eqref{eqn::langevindyn} robustly, with a particular method encoded by its string of characters. Given some mild assumptions on $\pi$ made explicit in \citep{leimkuhler2015computation}, we can utilize the Baker-Campbell-Hausdorff theorem  \citep{hairer2013geometric,leimkuhler2015computation}, to show that a method encoded by a symmetric string of these three pieces will give at least an $O(h^2)$ bias in observed averages,.



\subsection{Adding gradient noise}

Some consideration is needed when using a noisy gradient in Langevin dynamics. Swapping $\vec{F}$ for $\widetilde{\vec{F}}$ in \eqref{eqn::stepB} replaces the B `kick' step  
by a noisy update $\widetilde{\text B}$ where the force term is correct only in expectation. Our strategy is to use this noise term in place of some of the noise usually appearing in the O step and recover consistent sampling by choosing the damping matrix correctly in the O update.  To ensure that the gradient noise's covariance has sufficient rank, we inject an additional independent random term $\vec{R}\sim\N(\vec0,\lambda_h^2 \I)$ into the $\widetilde{\text B}$ step, where $\lambda_h$ is a positive constant chosen as a function of $h$. The $\widetilde{\text B}$ update function is therefore the sum of the stochastic gradient $\widetilde{\vec{F}}(\vec\theta)$ and the injected noise 
\[
\Phi^{\widetilde{\text B}}_{h,\widetilde{\vec Z}_h(\vec\theta)}(\vec\theta,\vec p)  = 
(\vec\theta , \vec p + h \widetilde{\vec{F}}(\vec\theta) + {\lambda_h} {\vec{R}} )= 
(\vec\theta , \vec p + h \vec{F}(\vec\theta) +   \widetilde{\vec Z}_h(\vec\theta) ) 
\]
where $\widetilde{\vec Z}_h(\vec\theta):=\lambda_h \vec{R} + h(\widetilde{\vec{F}}(\vec\theta)-{\vec{F}}(\vec\theta))$ includes all of the random terms, with  
\[\E[\widetilde{\vec Z}_h(\vec\theta)]= \vec0,  \quad {{\cov}[\widetilde{\vec Z}_h(\vec\theta)] = \vec{\Sigma}_{h}(\vec{\theta})=\lambda_h^2\I+h^2\vec{\Sigma}(\vec{\theta})}.\] 
It is important to note that although we may encode methods using the same splitting convention as in \eqref{eqn::ldpieces}, e.g. the A$\widetilde{\text B}$O$\widetilde{\text B}$A method, we no longer automatically recover second-order sampling (as we would normally for a palindromic string) due to the $\widetilde{\text B}$ term not correctly integrating the associated B piece defined in \eqref{eqn::ldpieces}.  

The $\widetilde{\text B}$ update we have defined is atomic (i.e. we cannot decouple the true force term from the noise term $\widetilde{\vec Z}_h(\vec\theta)$) so we are not able to directly manipulate the nuisance stochastic term in order to remove introduced bias.  
However, we may relate the noisy `$\widetilde{\text B}$O$\widetilde{\text B}$' update to the original `BOB' update by  comparing the resulting compositions. As we have injected noise already in the $\widetilde{\text B}$ step, we consider an O update using a damping matrix $\vec{\Gamma}_h$ and zero noise term.  Writing out the composed mappings, we have
\begin{align*}
\Phi^{\widetilde{\text B}}_{h/2,\widetilde{\vec Z}_{h/2}}\circ \Phi^{\text O}_{{\vec\Gamma}_h,\vec0} \circ\Phi^{\widetilde{\text B}}_{h/2,\widetilde{\vec Z}_{h/2}}(\vec\theta,\vec p) 
&= \left(\vec{\theta},{\vec\Gamma}_h\vec{p}+(\I+{\vec\Gamma}_h) (\tfrac{h}{2}\vec{F}(\vec\theta) +  \widetilde{\vec Z}_{h/2}(\vec\theta))  \right),\\
\Phi^{{\text B}}_{h/2}\circ \Phi^{\text O}_{{\vec \Gamma}_h,\vec R} \circ\Phi^{{\text B}}_{h/2}(\vec\theta,\vec p) 
&= \left(\vec{\theta},{\vec\Gamma}_h\vec{p}+\tfrac{h}2(\I+{\vec\Gamma}_h) \vec{F}(\vec\theta) +  \sqrt{\I - \vec{\Gamma}_h^2} \vec{R}  \right).
\end{align*}
We can reconcile these via the choice
\begin{align} 
{\vec\Gamma}_h \!
&= \! \left(\I-\vec \Sigma_{h/2}(\vec\theta) \right)\left(\I+ \vec \Sigma_{h/2}(\vec \theta) \right)^{-1} \!\nonumber \\
&= \! \left(\!\I-\lambda_{h/2}^2\I-\frac{h^2}4\vec \Sigma(\vec\theta) \! \right)\left(\!\I+\lambda_{h/2}^2\I+\frac{h^2}4\vec \Sigma(\vec \theta) \!\right)^{-1}\!\!\label{eqn::damping}
\end{align}
for which 
\[
\left(\I+{\vec\Gamma}_h\right) \left(\tfrac{h}{2}\vec{F}(\vec\theta) +  \widetilde{\vec Z}_{h/2}(\vec\theta)\right) = \tfrac{h}2\left(\I+{\vec\Gamma}_h\right) \vec{F}(\vec\theta) + \sqrt{\I-{\vec\Gamma}_h^2} \left(\sqrt{\vec{\Sigma}_{h/2}(\vec{\theta})}\right)^{-1}  \widetilde{\vec Z}_{h/2}(\vec\theta)
\]
and thus 
\begin{equation}
\Phi^{\widetilde{\text B}}_{h/2,\widetilde{\vec Z}_{h/2}}\circ \Phi^{\text O}_{{\vec\Gamma}_h,\vec0} \circ\Phi^{\widetilde{\text B}}_{h/2,\widetilde{\vec Z}_{h/2}}(\vec\theta,\vec p) = \Phi^{{\text B}}_{h/2}\circ \Phi^{\text O}_{{\vec\Gamma}_h,\widetilde{\vec{Y}}(\vec{\theta)}} \circ\Phi^{{\text B}}_{h/2}(\vec\theta,\vec p), \label{eqn::bobrel}
\end{equation}
where 
\[
\widetilde{\vec{Y}}(\theta):= \left(\sqrt{\vec{\Sigma}_{h/2}(\vec{\theta})}\right)^{-1}  \widetilde{\vec Z}_{h/2}(\vec{\theta})
\]
is a random vector with $\E[\widetilde{\vec{Y}}]=\vec0$ and $\cov(\widetilde{\vec{Y}})=\I$. This shows that for the particular choice of damping matrix \eqref{eqn::damping}, the `$\widetilde{\text B}$O$\widetilde{\text B}$' update is the same as the usual Langevin dynamics `BOB' update with random vector $\widetilde{\vec{Y}}$. Thus we propose an integrator coded A$\widetilde{\text B}$O$\widetilde{\text B}$A,
which we refer to  as the \NOGIN (NOisy Gradient INtegrator) scheme given explicitly in Algorithm \ref{alg::nogin}.  As this integrator contains a `$\widetilde{\text B}$O$\widetilde{\text B}$' string, we can use relation \eqref{eqn::bobrel} to relate its mapping to the second-order ABOBA scheme given in  \citep{leimkuhler2012rational}.   The noise term in the scheme, though not normally distributed in general,  will have zero mean and unit covariance which we demonstrate in Section \ref{sec::analysis} is enough to give a second order discretization of \eqref{eqn::langevindyn} for a specific choice of $\vec\mu(\vec\theta)$. 

\subsection{Efficient evaluation of the covariance} 

In principle the damping update we use (step \ref{step::damping} in Algorithm \ref{alg::nogin})  requires an evaluation of the exact $\vec\Sigma(\vec\theta)$ and the computation of a matrix inverse-vector multiply. As $N$ and $D$ get large this can seriously impact the computational efficiency of the method. This is not just a practical concern for \NOGIN, but all methods that utilize the stochastic gradient's covariance information share a similar challenge.

One solution to this problem is to utilize a low rank approximation to the covariance in place of the exact evaluation. For example, in the case of the estimator \eqref{eqn::robmun} we can write $\vec\Sigma$ as an outer product 
\[ 
\vec{\Sigma}(\vec{\theta}) = \frac{N-n}n\vec\Lambda(\vec{\theta}) \vec\Lambda^\T(\vec\theta), \qquad \vec\Lambda_{[:,\,i]}(\vec\theta) = \vec{F}_i(\vec{\theta}) - \vec{F}(\vec{\theta}),\qquad \vec\Lambda \in \R^{D\times N}
\]
where $\vec\Lambda_{[:,\,i]}$ is the $i^\text{th}$ column of matrix $\vec\Lambda$. Explicit evaluation of this term is computationally prohibitive as it requires computing all $N$ force terms $F_i$. Instead we may approximate the covariance by taking an outer product over the $n$ mini-batch force terms we compute at each step. If necessary we may subsample the estimate further to ensure that any matrix-vector calculations involving the estimate of $\Sigma$ can be done cheaply. For the inverse matrix-vector product, we may either use a conjugate gradient solver to approximate this rapidly, or use the Sherman-Morrison formula to find it explicitly. Similar approaches can be used if we build the covariance estimate from a weighted history of sparse outer products. Ultimately \NOGIN is no more expensive than other noisy gradient schemes leveraging covariance information and can be used in practice without explicitly forming the estimate for the  $\vec\Sigma$ matrix, instead representing it as a sparse vector product.

\subsection{Injected noise strength} 

We use a parameter $\lambda_h>0$ to adjust the strength of the injected noise to ensure that the noise covariance has full rank. This is closely related to the traditional choice of $\vec\mu=\gamma\vec\I$ in conventional Langevin dynamics. 

While in principal we may choose any positive value, for the choice of $\lambda^2_h = \tanh(\gamma h)$  the damping matrix \eqref{eqn::damping} has the desirable property that ${{\vec\Gamma}_h \to \exp(-h \gamma \I)}$ as ${\vec{\Sigma}(\vec{\theta})\to\vec0}$. As this is the usual form of damping matrix found in many Langevin dynamics applications, we will utilize this form for some fixed $\gamma>0$. However there are many choices for the variance of the injected noise, and we need not consider a scalar $\lambda_h$ at all. One alternative is to choose the injected noise such that $\vec\Sigma_h$ (the covariance of the random term in the noisy kick step) is kept at roughly a constant magnitude, though in practice this is computationally expensive and as it complicates the analysis we shall consider only scalar $\lambda_h$. 

\section{Error Analysis}\label{sec::analysis}
\subsection{Weak Convergence Analysis}
Given some initial conditions $(\vec{\theta},\vec{p})$, we denote the expected value of a smooth test function $f$ at time $t$ as
\[
u_f((\vec{\theta},\vec{p}),t) := \E\left[f(\vec{\theta}(t),\vec{p}(t))\,\middle|\,(\vec{\theta}(0),\vec{p}(0))=(\vec{\theta},\vec{p})\right],
\]
where the expectation is over all dynamical paths at time $t$. If the state evolves with respect to the underdamped Langevin dynamics \eqref{eqn::langevindyn}, then $u$ solves the backward Kolmogorov equation (see e.g. \citep{brooks2011handbook,leimkuhler2015molecular})
\begin{equation}
\frac{\partial u_f}{\partial t} = \left(\LL_{\text A}+\LL_{\text B}+\LL_{\text O}\right) u_f,\qquad u_f((\vec{\theta},\vec{p}),0) = f(\vec{\theta},\vec{p}). \label{eqn::kolmogorovld}
\end{equation}
If a numerical scheme integrating \eqref{eqn::langevindyn} has single-step stochastic update function $\Psi_h$ for  a timestep $h$, then define the expectation of $f$ after one step as $v_f$ where
\[
(\vec{\theta}_t,\vec{p}_t) = \Psi_h((\vec{\theta}_{t-1},\vec{p}_{t-1})),\qquad v_f((\vec{\theta},\vec{p}),h) := \E\left[f(\Psi_h((\vec{\theta},\vec{p})))\right].
\]
We assume that we may expand $v_f((\vec{\theta},\vec{p}),h)$ as a Taylor series in $h$ where 
\[
v_f((\vec{\theta},\vec{p}),h) = f(\vec{\theta},\vec{p}) + h ({\mathcal A}_1 f)(\vec{\theta},\vec{p}) + h^2 ({\mathcal A}_2 f)(\vec{\theta},\vec{p}) + \ldots 
\]
with linear operators ${\mathcal A}_i$ depending only upon $\nabla\log(\pi)$ and its derivatives and where the remainder term in the expansion can be bounded by a function independent of $h$ for  sufficiently small stepsize.

Given these assumptions, a scheme is  
weakly consistent to order $p$ if the Taylor series for $v_f$ and $u_f$ match to order $p+1$:
\begin{equation}
u_f((\vec{\theta},\vec{p}),h) - v_f((\vec{\theta},\vec{p}),h) = h^{p+1}( \LA f)(\vec{\theta},\vec{p}) + O(h^{p+1+s}) \label{eqn::consistency}
\end{equation}
for some $s>0$, and a nonzero linear differential operator $\LA$ depending smoothly on $\log(\pi)$ and its derivatives. We now give the main theoretical result of this article.
%
%
\begin{theorem} \label{thm::second}
Provided the given assumptions hold, the \NOGIN scheme is second-order weakly consistent with the dynamics \eqref{eqn::langevindyn} for sufficiently small $h$, where ${\vec{\mu}(\vec{\theta})=\gamma\I+h\vec\Sigma(\vec\theta)}/2$.
\end{theorem}
%
%
\begin{proof}
In the case of the \NOGIN scheme we have a single-step update as the composition of maps
\begin{align*}
\Psi_h((\vec\theta,\vec{p})\,|\,\widetilde{\vec{Z}} ) = \Phi^{{\text A}}_{h/2} \circ\Phi^{\widetilde{\text B}}_{h/2,\widetilde{\vec{Z}}} \circ \Phi^{\text O}_{{\vec\Gamma}_h,\vec0} \circ\Phi^{\widetilde{\text B}}_{h/2,\widetilde{\vec{Z}}} \circ \Phi^{{\text A}}_{h/2}(\vec\theta,\vec{p})
\end{align*}
where $\widetilde{\vec{Z}}(\vec{\theta})=\lambda_{h/2} \vec{R} + \tfrac{h}2(\widetilde{\vec{F}}(\vec\theta)-{\vec{F}}(\vec\theta))$ represents the injected and gradient noise and where ${\vec\Gamma}_h$ is chosen as in \eqref{eqn::damping}, with $\lambda^2_h = \tanh(\gamma h)$ for some $\gamma>0$. 
For this choice of ${\vec\Gamma}_h$ one can extricate the force updates using \eqref{eqn::bobrel} to rewrite the step as 
\[
\Psi_h((\vec\theta,\vec{p})\,|\,\widetilde{\vec{Z}}) =\Phi^{{\text A}}_{h/2} \circ\Phi^{{\text B}}_{h/2} \circ \Phi^{{\text O}}_{{\vec\Gamma}_h,\widetilde{\vec{Y}}} \circ\Phi^{{\text B}}_{h/2} \circ \Phi^{{\text A}}_{h/2}(\vec\theta,\vec{p}),
\quad 
\widetilde{\vec{Y}}(\vec{\theta}) = \sqrt{\cov(\widetilde{\vec{Z}}(\vec{\theta}))^{-1}}\widetilde{\vec{Z}}(\vec{\theta})
\]
where $\widetilde{\vec{Y}}(\vec{\theta})  $ is a random vector with $\E[\widetilde{\vec{Y}}]=\vec0$ and $\cov[\widetilde{\vec{Y}}]=\I$. We now relate this O update to the correct mapping integrating exactly the O piece in \eqref{eqn::stepO}.    For the choice of ${\vec{\mu}(\vec{\theta})=\gamma\I+h\vec\Sigma(\vec\theta)}/2$ and sufficiently small $h$, we have
\[
{\vec\Gamma}_h = \vec{M}_h+O(h^{7/2}),\qquad \sqrt{\I-\tilde{\vec\Gamma}_h^2} = \sqrt{\I-\vec{M}_h^2}+O(h^{7/2}).
\]
where $\vec{M}_h=\exp(-\vec\mu(\vec\theta)h)$. 
Thus we have 
\begin{equation}\label{eqn::proof1}
f\left(  \Phi^{{\text O}}_{{\vec\Gamma}_h,\widetilde{\vec{Y}}}(\vec{\theta},\vec{p}) \right) 
=f\left(  \Phi^{{\text O}}_{
\vec{M}_h
,\widetilde{\vec{Y}}}
(\vec{\theta},\vec{p}) \right) + O(h^{7/2}).
\end{equation}
The random vector $\widetilde{\vec Y}$ is not $\N(\vec0,\vec\I)$ due to the noisy force term, however as we only need it to weakly approximate the O step to a given order it is enough that it has zero mean and unit covariance. We can see this by expanding $\widetilde{\vec{Y}}$ in powers of $h$ and writing $\vec{S}:=\widetilde{\vec{F}}(\vec\theta)-{\vec{F}}(\vec\theta)$, where we have 
\begin{align*}
\widetilde{\vec{Y}} &= \sqrt{\left(\tanh(\gamma h/2)\I +\frac{h^2}4\vec\Sigma\right)^{-1}}\left(\sqrt{\tanh(\gamma h/2)} \vec{R} + \frac{h}2\vec{S}\right)\\
&= \sqrt{\left(\I +\frac{h^2}{4\tanh(\gamma h/2)}\vec\Sigma\right)^{-1}}\left( \vec{R} + \frac{h}{2\sqrt{\tanh(\gamma h/2)}}\vec{S}\right)\\
&= \left(\I - \frac{h}{4\gamma}\vec\Sigma + \frac{3h^2}{32\gamma^2}\vec\Sigma^2\right)\vec{R}+ \left(\sqrt{\frac{h}{2\gamma}} - \sqrt{\frac{h^3}{32\gamma^3}}\vec\Sigma \right) \vec{S}+O(h^{5/2}) = \vec{C}_1 \vec{R} + \vec{C}_2 \vec{S} + O(h^{5/2})
\end{align*} 
for matrices $\vec{C}_1$ and $\vec{C}_2$. Replacing $\widetilde{\vec{Y}}$ by its expansion in \eqref{eqn::proof1} and expanding in the $\vec{S}$ terms gives
\begin{align*}
f\left(  \Phi^{{\text O}}_{{\vec\Gamma}_h,\widetilde{\vec{Y}}}(\vec{\theta},\vec{p}) \right)&=
f\left(  \vec{\theta}, \vec{M}_h \vec{p} + \sqrt{\I-\vec{M}_h^2} \vec{C}_1  {\vec{R}} + \sqrt{\I-\vec{M}_h^2} \vec{C}_2 \vec{S} \right) + O(h^{3}),\\
&= f\left(  \vec{\theta}, \vec{M}_h \vec{p} + \sqrt{\I-\vec{M}_h^2} \vec{C}_1  {\vec{R}}\right)\\
&\quad+ \sqrt{\I-\vec{M}_h^2} \vec{C}_2 \vec{S} \cdot \nabla_p f\left(  \vec{\theta}, \vec\Gamma_h \vec{p} + \sqrt{\I-\vec{M}_h^2} \vec{C}_1  {\vec{R}}\right) \\
&\quad+ \frac{h^2}{2} \vec{S}\vec{S}^\T : \nabla_p^2 f\left(  \vec{\theta}, \vec{M}_h \vec{p} + \sqrt{\I-\vec{M}_h^2} {\vec{R}}\right)
+ O(h^{3})
\end{align*}
Taking expectations and using the independence of $\vec{R}$ and $\vec{S}$ leads to the linear terms dropping out, with
\[
\E\left[f\left(  \Phi^{{\text O}}_{{\vec\Gamma}_h,\widetilde{\vec{Y}}}(\vec{\theta},\vec{p}) \right)\right]= \E\left[f\left(  \vec{\theta}, \vec{M}_h \vec{p} + \sqrt{\I-\vec{M}_h^2} \vec{C}_1  {\vec{R}}\right) \right]+ \frac{h^2}{2}\E\left[ \vec{\Sigma} : \nabla_p^2 f\left(  \Phi^{{\text O}}_{{\vec{M}_h},{\vec{R}}}(\vec{\theta},\vec{p}) \right) \right]
+ O(h^{3}).
\]
Expanding this first term 
yields
\begin{align*}
\E\left[f\left(  \vec{\theta}, \vec{M}_h \vec{p} + \sqrt{\I-\vec{M}_h^2} \vec{C}_1  {\vec{R}}\right) \right]&=
\E\left[f\left(  \vec{\theta}, \vec{M}_h \vec{p} + \sqrt{\I-\vec{M}_h^2} {\vec{R}}\right) \right]\\
&\quad- \frac{h^2}{2}\E\left[ \vec{R}\vec{\Sigma}\vec{R}^\T : \nabla_\vec{p}^2 f\left(  \Phi^{{\text O}}_{\vec{M}_h,{\vec{R}}}(\vec{\theta},\vec{p}) \right) \right] + O(h^3),
\end{align*}
and hence we obtain
\begin{align*}
\E\left[f\left(  \Phi^{{\text O}}_{{\vec\Gamma}_h,\widetilde{\vec{Y}}}(\vec{\theta},\vec{p}) \right)\right] &= 
\E\left[f\left(  \Phi^{{\text O}}_{\vec{M}_h,{\vec{R}}}(\vec{\theta},\vec{p}) \right)\right] + O(h^3) \\
&= \left(e^{h\LL_{\text O}} f\right)(\vec{\theta},\vec{p}) + h^3 ( {\mathcal A} f)(\vec{\theta},\vec{p}) +  O(h^{7/2})
\end{align*}
for operator ${\mathcal A}$ depending upon $\log(\pi)$ and its derivatives, with $\vec{M}_h=\exp(-h\vec{\mu}(\vec\theta))$ and ${\vec{R}}\sim\N(\vec0,\I)$.  Thus   the single-step expectation is
\[
v_f((\vec{\theta},\vec{p}) ,h) =
\left(e^{h\LL_{\textrm A}/2}e^{h\LL_{\textrm B}/2}e^{h\LL_{\textrm O} }e^{h\LL_{\textrm B}/2}e^{h\LL_{\textrm A}/2} f \right)(\vec\theta,\vec{p}) + O(h^3).
\]
This can be written as 
\[
v_f((\vec{\theta},\vec{p}) ,h) =\left(e^{h(\LL + h^2 {\mathcal{X}})} f \right)(\vec\theta,\vec{p}) + O(h^3)=u_f((\vec{\theta},\vec{p}),h) + O(h^3)
\]
where the operator ${\mathcal{X}}$ is explicitly given through the Baker-Campbell-Hausdorff (BCH) formula \citep{hairer2013geometric}.
\end{proof}
 
Using the results of \cite[Theorem 4.3]{abdulle2014high}, it can be shown that second order schemes provide an $O(h^m)$ infinite-time bias, where $m\geq2$. Indeed in Section \ref{sec::gaussianexact} we demonstrate that the \NOGIN scheme gives no infinite-time bias when considering Gaussian $\pi$.

\subsection{Exactness For Gaussian Distributions} \label{sec::gaussianexact}
In the case of normally distributed gradient noise we recover an exactness result for \NOGIN in terms of introduced bias into the target distribution. This comes from a stronger result giving the perturbed invariant distribution preserved by the numerical scheme.
\begin{lemma}\label{lem:exactness}
If the target distribution is of the form  ${\pi(\vec\theta,\vec p) = \N(\vec \theta\,|\,\vec\eta,\vec\Omega) \times \N(\vec p\,|\,\vec0,\I)}$ for positive definite $\vec\Omega$, then for $h^2<4\kappa$ where $\kappa$ is the largest eigenvalue of $\vec\Omega$, and gradient noise $\widetilde{\vec{F}}(\vec{\theta}) \sim \N(\vec{F}(\vec\theta),\vec\Sigma(\vec\theta))$ with positive definite $\vec{\Sigma}$, the \NOGIN scheme preserves the perturbed distribution
\[
{\pi}_h(\vec\theta,\vec p) = \N(\vec\theta\,|\,\vec\eta,\vec\Omega) \times \N\left(\vec p\,\middle|\,\vec 0,\left(\I-\frac{h^2}4\vec\Omega^{-1}\right)^{-1} \right).
\]
\end{lemma}

\begin{proof}
This can be demonstrated directly using the update maps in (\ref{eqn::stepB}-\ref{eqn::stepO}). Plugging in the mappings, we have
\begin{align*}
 \left.\Phi^{{\text B}}_{h/2}\circ\Phi^{{\text A}}_{h/2}\right.(\vec\theta,\vec p)  &\sim \pi_h'(\vec\theta,\vec p)\quad \text{if} \quad (\vec\theta,\vec p) \sim {\pi}_h  (\vec\theta,\vec p)\\ 
 \left.\Phi^{{\text A}}_{h/2}\circ\Phi^{{\text B}}_{h/2}\right.(\vec\theta,\vec p)  &\sim \pi_h(\vec\theta,\vec p)\quad \text{if} \quad (\vec\theta,\vec p) \sim {\pi}'_h  (\vec\theta,\vec p)  
\end{align*}
for the distribution
\[
\pi_h'(\vec\theta,\vec p) := \N\left(\vec\theta\,\middle|\,\vec\eta,\vec\Omega\left(\I-\frac{h^2}4\vec\Omega^{-1}\right)^{-1}\right) \times \N\left(\vec p\,\middle|\,\vec 0,\I\right).
\]
As we choose the damping matrix carefully to balance the noise correctly, and as the gradient noise is normally distributed, the O solve amounts to an exact weak solve of an Ornstein-Uhlenbeck process. Thus it preserves the correct marginal distribution in $\vec{p}$, and so
\[
\Phi^{{\text O}}_{\tilde{\vec{\Gamma}}_h,\widetilde{\vec{Y}}} (\vec\theta,\vec p) \sim {\pi}_h' (\vec\theta,\vec p) \quad \text{if} \quad (\vec\theta,\vec p) \sim {\pi}_h' (\vec\theta,\vec p).
\]
Hence the proposed distribution is preserved at each step
\[
 \Phi^{{\text A}}_{h/2}\circ\Phi^{{\text B}}_{h/2}\circ\Phi^{{\text O}}_{h,\widetilde{\vec{Y}}}\circ\Phi^{{\text B}}_{h/2}\circ\Phi^{{\text A}}_{h/2}(\vec\theta,\vec p)  \sim \pi_h(\vec\theta,\vec p) \quad \text{if} \quad (\vec\theta,\vec p) \sim {\pi}_h (\vec\theta,\vec p)
\]
as required. 
\end{proof}

Taking the marginal of $\pi_h$ over $\vec{\theta}$, we can see that the correct Gaussian target distribution $\N(\vec \theta\,|\,\vec\eta,\vec\Omega)$ is recovered without discretization bias. Given ergodicity (which can be proven rigorously using the machinery from e.g. \citep{leimkuhler2015computation}), this gives the unique long-time distribution for trajectories generated using the \NOGIN scheme.

\subsection{Diminishing returns on reducing $n$} \label{sec::zerosum}

In terms of the introduced bias, Lemma \ref{lem:exactness}  shows that for Gaussian $\pi$ we  suffer no consequences sampling using \NOGIN with normally distributed gradient noise with a large covariance (equivalently taking $n$ as small as we wish).  However there is a tradeoff in terms of sampling efficiency  as the gradient noise increases.  

This  is easy to see from the definition of \NOGIN.  The damping matrix in  step \ref{step::damping} in Algorithm \ref{alg::nogin} will appropriately damp the momentum to preserve the correct distribution, with a large gradient noise requiring a large damping that suppresses mixing. 
%
We can quantify the mixing rate using the integrated autocorrelation time (IAT) \citep{SokalIAT}, denoted $\tau$. For a trajectory $(\theta_t)$ and a suitable test function $f$ the observed error in computed averages behaves like
\[
\lim_{T\to\infty} T\,\E\left[\left(\bar{f}_T - \left< f \right>\right)^2\right] = { \sigma_f^2 \tau_f } ,
\qquad 
\bar{f}_T := \frac1T \sum_{t=1}^T f(\theta_t)
\]
where $\bar{f}_T$ is the observed average after $T$ steps and $\sigma_f^2$ is the variance of $f(\theta)$ under $\pi$ unrelated to the numerical method. The IAT $\tau_f$ can be thought of as the marginal number of timesteps required to generate an additional independent sample.  

We consider  applying the \NOGIN scheme to a one-dimensional standard normal distribution $\pi(\theta)=\N(\theta\,|\,0,1)$ with a normally distributed stochastic gradient $\widetilde F(\theta)\sim \N(F(\theta),\sigma^2)$ for variance $\sigma^2$ independent of $\theta$ and $h$. We may write an update of the \NOGIN scheme as
\[
\left[\begin{array}{c} \theta_{k+1}\\p_{k+1}\end{array}\right] = 
\vec{A} \left[\begin{array}{c} \theta_{k}\\p_{k}\end{array}\right] + 
\frac12 (1+\Gamma_h) h \sigma  R_k\left[\begin{array}{c} h/2 \\1\end{array}\right] ,
\]
for the scalar damping constant $|\Gamma_h|<1$ given in  \eqref{eqn::damping}, and $R_k\sim \N(0,1)$  with 
\[
\vec{A} = \frac18 \left[\begin{array}{cc} 8-2h^2(1+\Gamma_h) & (4h-h^3)(1+\Gamma_h)\\-4(1+\Gamma_h)h & 8\Gamma_h-2h^2(1+\Gamma_h)\end{array}\right].
\]
If $\lambda$ is the largest eigenvalue of $\vec{A}$ then its  corresponding eigenvector $\vec{v}$ is the direction in which the system mixes the slowest.  Thus the worst-case mixing rate of the system (among linear observables) can be found by computing the integrated autocorrelation time of ${f(\vec{z})=\vec{z}\cdot\vec{v}}$ as we change $\sigma$ with fixed $h$, where $\vec{z}=[\theta,p]^\T$. The autocorrelation function for $f$ at time lag $k$ is
\[
\textrm{acf}_f(k) 
:= \frac{ \E[(\vec{z}_0 \cdot \vec{v}) (\vec{z}_k \cdot \vec{v})]}{\E[(\vec{z}_0 \cdot \vec{v}) (\vec{z}_0 \cdot \vec{v})]} 
= \frac{ \E\left[(\vec{z}_0 \cdot \vec{v}) (\vec{z}_0 \cdot (\vec{A}^\T)^k\vec{v})\right]}{\E[(\vec{z}_0 \cdot \vec{v}) (\vec{z}_0 \cdot \vec{v})]} = \lambda^k,
\]
where the expectation is over all initial conditions weighted according to the known invariant distribution. The IAT $\tau_f$ is then 
\[
\tau_f := 1 + 2\sum_{k=1}^\infty \textrm{acf}_f(k) =
\frac{1+\lambda}{1-\lambda}. 
\]
From \eqref{eqn::damping} the damping term is  $\Gamma_h=(1-\sigma^2h^2/4)/(1+\sigma^2h^2/4)$. Plugging this value into $\vec A$  and considering large $\Sigma$ gives  $\lambda=1-2\sigma^{-2}+O(\sigma^{-4})$. Hence
\[
\tau_f \approx \sigma^2
\]
when the variance of the gradient noise is large. This is not an effect of the bias; Lemma \ref{lem:exactness} still holds so we expect bias to be independent of $\sigma$. Instead, the required damping inhibits the exploration of the system.

If we use an estimator such as \eqref{eqn::robmun} where the covariance scales like $C^2=O(1/n)$ then for sufficiently small $n$ we do not gain an efficiency boost by reducing $n$ further as we proportionately increase the IAT. So we would expect that halving the batchsize would double the number of steps needed to decorrelate, keeping the observed error     constant for a fixed amount of computation $Tn$.

\section{Numerical Experiments}\label{sec::experiments}
We compare the sampling bias and efficiency between our proposed \NOGIN method and several other  state-of-the-art schemes: the SGLD method \citep{welling2011bayesian};  the modified-SGLD method (mSGLD)  \citep{vollmer2016exploration}; the SGNHT scheme \citep{Leimkuhler2016}; the SGHMC scheme \citep{chen2014stochastic};  and the CC-ADL scheme \citep{shang2015covariance}. Algorithms are implemented in a python code available online\footnote{Code can be downloaded at {\tt{https://github.com/c-matthews/nogin}}. Results are generated using repository code from March 2019}.

All schemes, apart from  SGLD and SGNHT, make use of covariance information to (attempt to) remove introduced bias. In the following experiments we approximate $\vec\Sigma$ from a finite weighted history of previous force terms. The same approximated covariance information is used for all methods that utilize it.

We run multiple experiments using different batch sizes $n$ and a fixed stepsize $h$. We compare the mean-squared error (MSE) in the computed variance, as a function of the number of complete passes through the data, known as the \emph{epoch}. The MSE is compared to the exact result computed via a long simulation of the unbiased MALA scheme \citep{roberts1996}.  Decreasing the batch size enables us to take more steps per entire pass through the dataset, so we would expect the accuracy to increase as $n$ decreases up until the sampling error is swamped by the bias from the noisy gradient.  We consider the computation of force terms $\vec{F}_i$ to be the principle cost of the sampling, so we will use the number of epochs to measure total computational effort.

\subsection{Gaussian mixture model} \label{sec::gmm}
\begin{figure}[t]
\centering
\includegraphics[trim={0 2 0 0},clip,width=\textwidth]{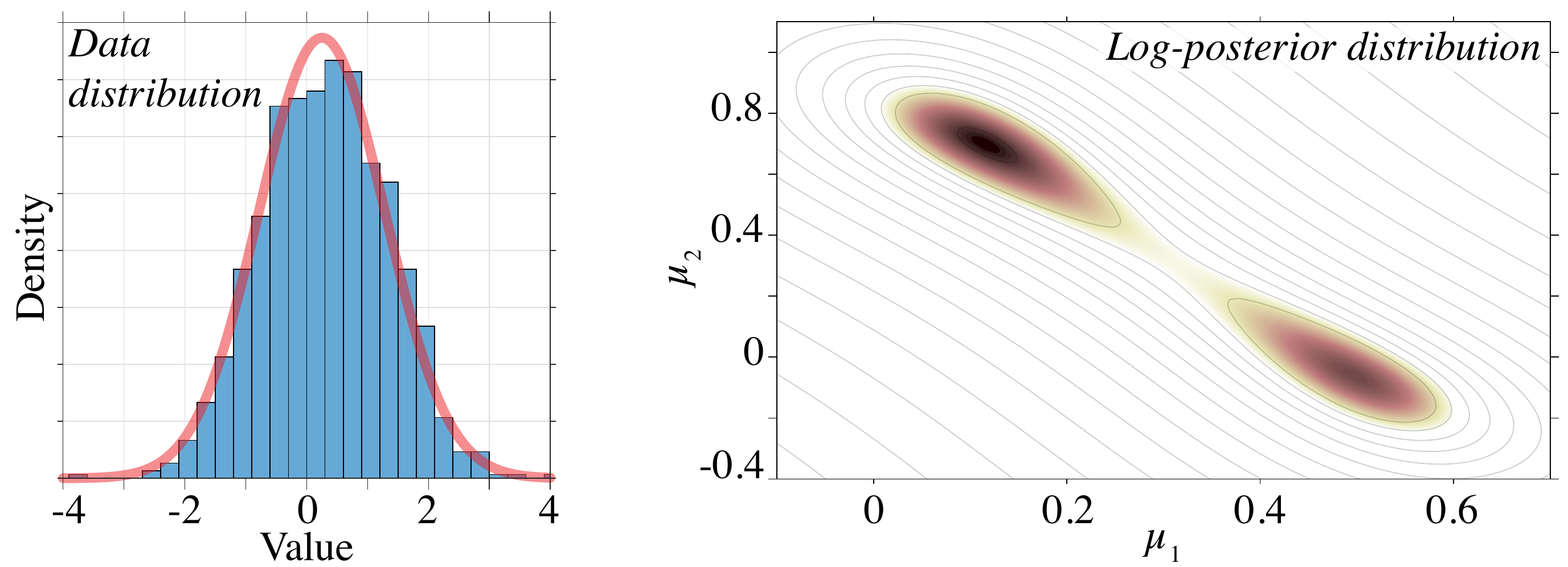}
\caption{ Left: The histogram plotting the density of the 1000 datapoints $\vec y$, with the red line showing the true distribution from which they were drawn. Right: The two-dimensional log-posterior function $\log\pi(\vec\theta)$ for $\vec\theta=[\mu_1,\mu_2]^\T$. }
\label{fig::gmix1}
\end{figure}

We consider one-dimensional data distributed according to a  Gaussian mixture distribution with mixture centres $\mu_1$ and $\mu_2$ where
\[
\pi(y\,|\,\vec\theta) \propto  \exp\left(-(y-\mu_1)^2/2\right) + 2 \exp\left(-(y-\mu_2)^2/2\right),\qquad \vec\theta=[\mu_1,\mu_2]^\T.
\]
We draw a dataset $\vec y=\{y_i\}$ of $N=1000$ datapoints, where $y_i\sim\pi(y_i\,|\,\vec\theta^*)$ with $\vec\theta^*=[0.5,0]^\T$. Choosing a uniform prior $\pi_0$, we sample points from the posterior distribution given in  \eqref{eqn::target_distribution} and plotted in Figure \ref{fig::gmix1}. It is clear that the data imply a bimodal posterior distribution, with the more unlikely basin containing $\vec\theta^*$.

For each algorithm we compare the convergence of the MSE of the variance of $\vec\theta$, using a batch size $n$ over the course of $30,000$ epochs. Using a smaller $n$ means we are able to take more steps for a fixed number of epochs, although a small batch size risks increasing the bias from the gradient noise. The results of the experiments are plotted in Figure \ref{fig::gmix2}. 

We choose the stepsize $h$ so that all algorithms perform equally when using the true gradient, as is visible in the results when $n=1000$. This ensures that we are comparing performance of the algorithms fairly as we reduce the batch size.

Of the methods tested, it is clear that \NOGIN provides the smallest error for the fewest passes through the data. It is the only scheme  that was able to provide errors below $10^{-6}$ within the allowed number of epochs. It is also the only scheme that continues to see improvement below a batch-size of $n=10$, although discretization error is present as the batch size continues to decrease.

The SGNHT method performs very poorly in this test, possibly due to the strong position-dependence of the force's covariance matrix (due to the bimodality of $\pi$). Consistency of the SGNHT method relies upon a constant covariance matrix as it does not use information about $\vec\Sigma$, so the results are not unexpected. By contrast, SGLD also does not use covariance information, but it performs significantly better than SGNHT as we lower $n$. Though the bias in SGLD is large it performs reliably even for small $n$. The mSGLD method reduces the bias for large $N$ but performs worse than SGLD when $n$ is smaller.

The SGHMC scheme requires us to manually increase the friction in order to maintain fluctuation-dissipation. The resulting diffusion through the landscape is slower, so that although each step is cheaper we end up with no improvement in efficiency.

\begin{figure}[ht]
\centering
\includegraphics[trim={0 0 0 0},clip,width=\textwidth]{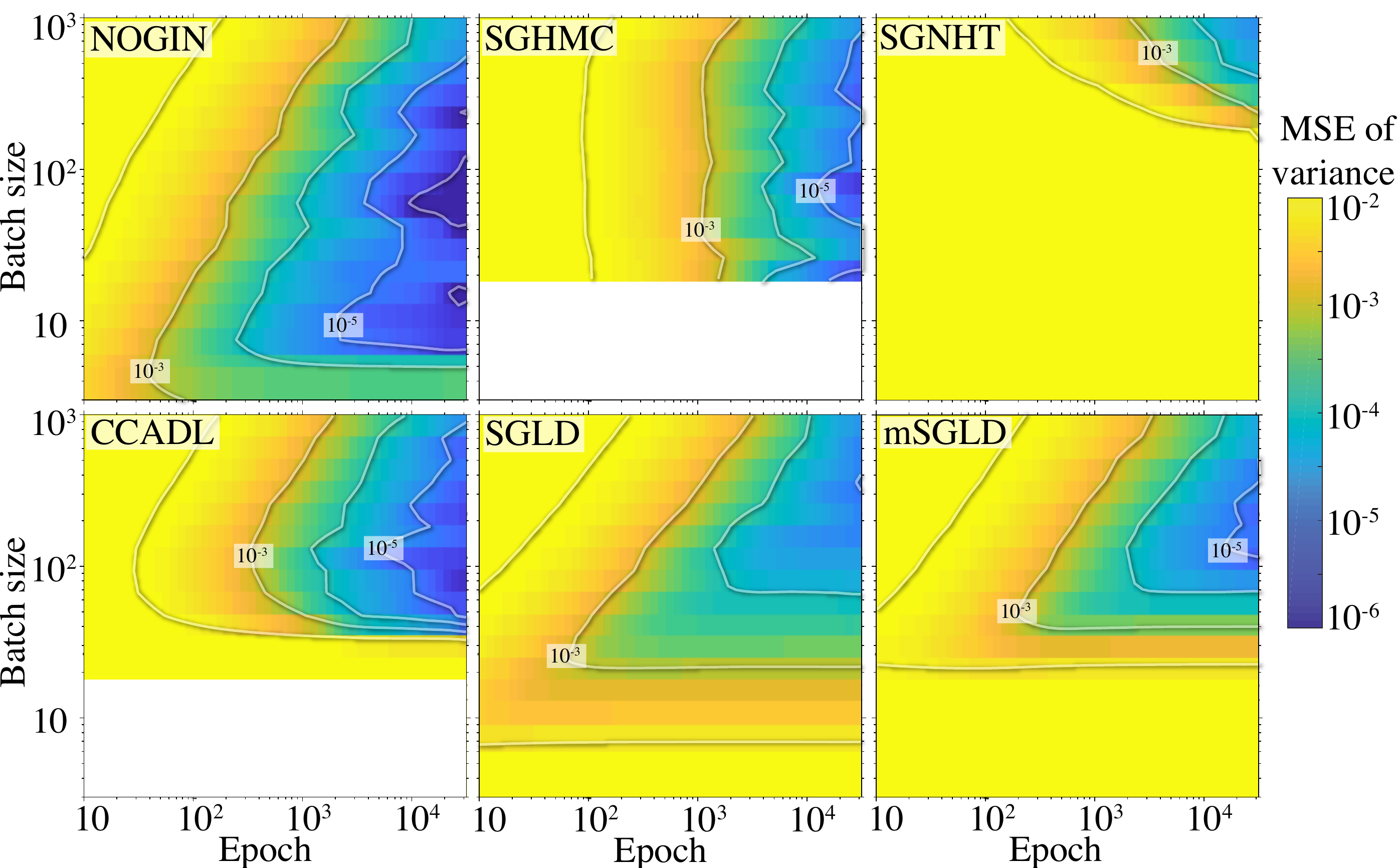}
\caption{ We plot the MSE in the computed variance for the Gaussian mixture experiment. White pixels indicate instability. }
\label{fig::gmix2}
\end{figure}

\subsection{Bayesian Logistic Regression}
We use Bayesian Logistic Regression (BLR) for classification by fitting feature parameters $\vec{\theta}\in\R^D$ to the dataset $\vec{y}=\{\vec{y}_i\}$ containing $N$ records $\vec{y_i}=\{\vec{x}_i,c_i\}$ where $\vec{x}_i\in\R^{D}$ is a feature vector and $c_i\in\{0,1\}$ is a binary class label indicating item $i$'s classification. We then use the standard BLR likelihood
\[
\pi(\vec y_i\,|\,\vec\theta) = \left(\frac1{1+e^{-\vec \theta\cdot \vec x_i}}\right)^{c_i} \left(1-\frac1{1+e^{-\vec \theta\cdot \vec x_i}}\right)^{(1-c_i)},
\]
with a prior $\pi_0(\vec\theta) = \N(\vec\theta\,|\,\vec0, 100 \I)$.  We run a BLR experiment for classifying the number 7 or 9 in the MNIST dataset. Using principal component analysis on the dataset we project the data $\vec{x}$ onto the top 128 produced eigenvectors. Including the constant term this gives $D=129$, with a total dataset size of $N=12214$. Just as in Section \ref{sec::gmm} we test all six algorithms using  the   estimator \eqref{eqn::robmun} to compute a noisy force, with the covariance estimated from a weighted finite history of previously computed force terms. We look at the mean squared error in the the variance as a function of total passes through the dataset. The results are plotted in Figure \ref{fig::blr} running at a fixed stepsize $h$ for each method. 

It is clear that the \NOGIN scheme provides superior sampling efficiency compared to the other schemes, giving around a $1\%$ error from using two hundred passes through the dataset. Additionally, the method remains stable when using e.g. $n=100$ samples per evaluation of $\widetilde{\vec{F}}$, which the other schemes cannot deliver at the chosen stepsize.

\begin{figure}[tb]
\centering
\includegraphics[trim={0 0cm 0 0},clip,width=\textwidth]{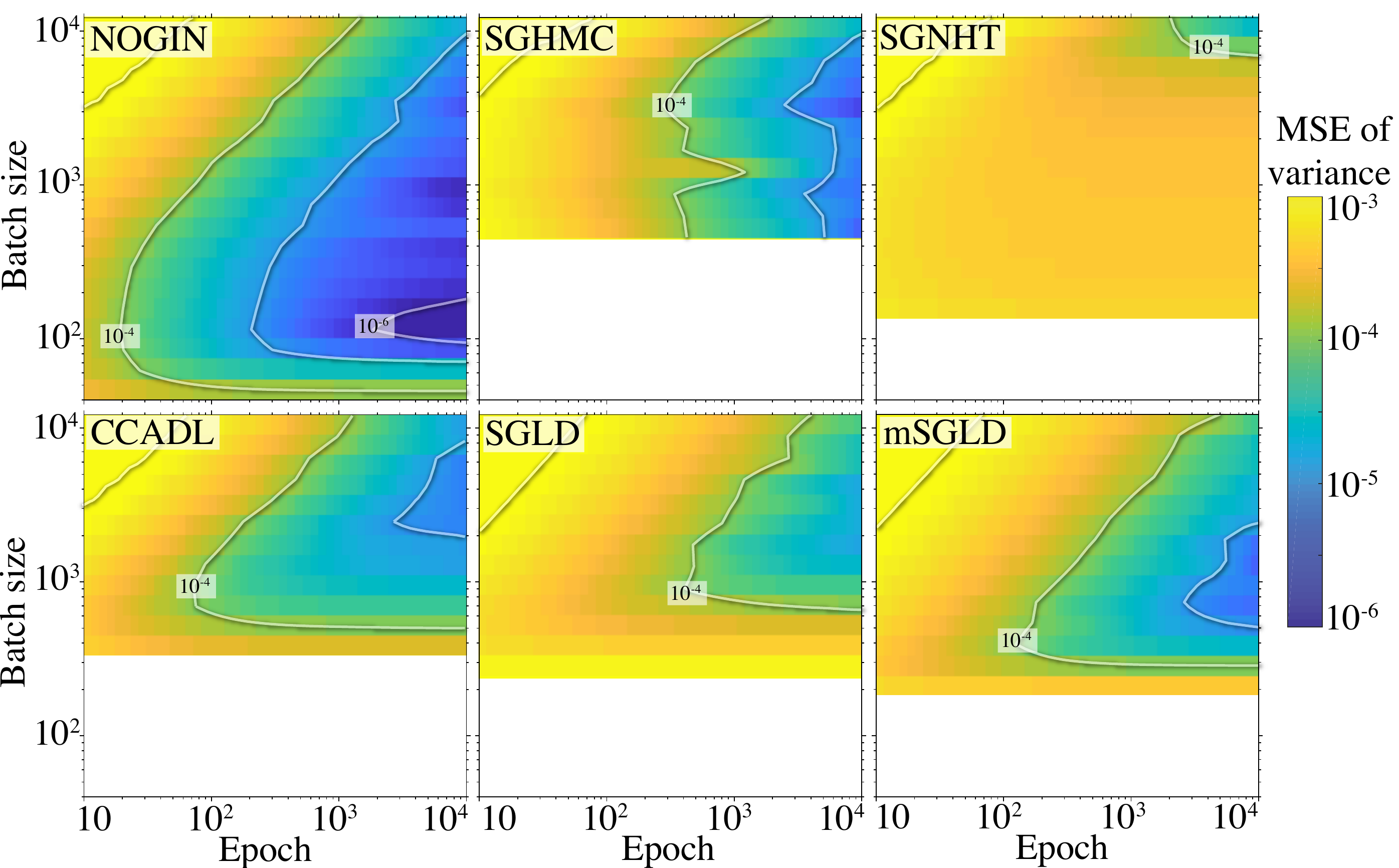}
\caption{Results for the BLR model using the MNIST 7 and 9 dataset of 12214 data points. The error in the computed parameter variance is given as a function of the number of passes through the entire dataset (epochs). 
}
\label{fig::blr}
\end{figure}
 
The  \NOGIN scheme plateaus in efficiency for $n<300$ in this problem. The vertical contour demonstrates that the efficiency gain by running with a smaller $n$ is offset by the damping term that reduces exploration rate by the same factor. This is the effect shown from the analysis in Section \ref{sec::zerosum}. This effect can be mitigated by using a more accurate estimator, which reduces the gradient noise and hene the size of the required damping. There may be other benefits for using a smaller $n$ beyond  computational efficiency, for example memory restrictions, that make the enhanced stability of \NOGIN in this regime beneficial.

\section{Discussion and Conclusion}\label{sec::conclusion}
In this paper we have presented a novel MCMC sampling algorithm for  using   noisy gradient information to sample from a prescribed target distribution. We demonstrate that the \NOGIN scheme introduces a bias to computed averages that scaled as the second-order in the stepsize $h$ in the general case, while remaining stable even when using a small mini-batch size. If the gradient noise is normally distributed, then the scheme preserves the exact distribution when sampling from a quadratic log-posterior with sufficiently small timestep. In numerical tests the scheme provides significant improvements over other stochastic gradient schemes, both in stability and in accuracy.  

Our  analysis demonstrates that efficiency plateaus when the gradient noise is large, due to a correspondingly large damping required to balance the fluctuation-dissipation relation, which slows mixing. Interestingly this slowdown is exactly countered by the reduced cost of the force and so sampling efficiency remains constant in this regime. Applying \NOGIN with a more accurate estimator  than \eqref{eqn::robmun} would reduce the variance of the gradient and alleviate this issue, which we leave to further work.

 
\bibliographystyle{abbrv}      
\bibliography{refs}   

\end{document}